\documentclass[reqno]{amsart}
\usepackage{amsmath,amsthm,pdfsync, amssymb, accents}
\usepackage{graphicx,comment}
\usepackage{pgfplots}
\usepackage{pdfsync}
\usepackage{xr}
\usepackage{color}
\usepackage{accents}
\usepackage{epic}
\usepackage{eucal}
\usepackage{cases}

\numberwithin{equation}{section}
\numberwithin{figure}{section}

\usepackage{amsmath}
\usepackage{float}
\usepackage{amsfonts}
\usepackage{amssymb}
\usepackage{fancyhdr}
\usepackage{graphicx}
\usepackage{amsthm}
\usepackage{tikz,xifthen}
\usepackage{soul}

\newcommand{\wb}[1]{\overline{#1}}
\newcommand{\ub}[1]{\underline{#1}}

\def\m@th{\mathsurround=0pt}
\mathchardef\bracell="0365
\def\upbrall{$\m@th\bracell$}
\def\undertilde#1{\mathop{\vtop{\ialign{##\crcr
				$\hfil\displaystyle{#1}\hfil$\crcr
				\noalign
				{\kern1.5pt\nointerlineskip}
				\upbrall\crcr\noalign{\kern-5pt
	}}}}\limits}

\newcommand{\al}{\alpha}
\newcommand{\bt}{\beta}
\newcommand{\gm}{\gamma}

\newcommand{\dl}{\delta}

\newcommand{\lm}{\lambda}

\newcommand{\bblu}{\begin{color}{blue}}
	\newcommand{\bred}{\begin{color}{red}}{\tiny }
		\newcommand{\ecl}{\end{color}}

	\newtheorem{thm}{Theorem}
	
	\newtheorem{prp}{Proposition}

	\numberwithin{thm}{section}
	\numberwithin{prp}{section}
	
\begin{document}
		
		\title{Initial Value Spaces of Integrable Lattice Equations}
		
		\author{Nalini Joshi}\thanks{This research was supported by an Australian Laureate Fellowship \# FL 120100094 from the Australian Research Council. }
		\address{School of Mathematics and Statistics F07, The University of Sydney, NSW 2006, Australia}
		\email{nalini.joshi@sydney.edu.au}
		\author{Sarah B. Lobb}
		\address{School of Mathematics and Statistics F07, The University of Sydney, NSW 2006, Australia}
		\email{sarah.lobb@gmail.com}
		\author{Matthew Nolan}\thanks{Supported by an Australian Postgraduate Award.}
		\address{School of Mathematics and Statistics F07, The University of Sydney, NSW 2006, Australia}
		\email{m.nolan@maths.usyd.edu.au}
		
		\begin{abstract}
			In this paper, we examine the space of initial values for integrable lattice equations, which are lattice equations classified by Adler {\em et al} (2003), known as ABS equations. By considering the map which iterates the solution along particular directions on the lattice, we perform resolutions of singularities for several examples of ABS equations for the first time. Our geometric observations lead to new Miura transformations and reductions to ordinary difference equations.
			\begin{center}
				{\tiny{\today}}
			\end{center}
		\end{abstract}
		\maketitle
		\section{Introduction}\label{sect_intro}
		The primary goal of this paper is to construct an analogue of initial value spaces for integrable partial difference equations, naturally following the constructions of such spaces by Sakai\cite{Sak2001,KNY2016} for ordinary difference equations. To be specific, we focus on four examples of ABS equations \cite{AdlBobSur2003}, embedding them in projective space $\mathbb P^3$ and regularising the space by resolving singularities. Although the sequence of resolutions is guaranteed to terminate by standard theorems in algebraic geometry, the sequence is no longer algorithmic due to the difficulties and richness of higher dimensions. Nevertheless, we are able to carry out the resolutions. As byproducts, we derive new transformations and reductions of these equations.
		
		Our starting point is a quadrilateral as shown in Figure \ref{fig_square}.
		\begin{figure}[H]
			\begin{center}
				\begin{tikzpicture}[scale = 2]	
				\draw[-] (1.1,1) -- (-0.1,1) node[above] {$v$};
				\draw[-] (1,-0.1) -- (1,1.1) node[right] {$y$};
				\draw[-] (0,1.1) -- (0,-0.1) node[left] {$x$};
				\draw[-] (-0.1,0) -- (1.1,0) node[below] {$u$};
				\draw[black] (0.5,-0.1) node {$\al$};
				\draw[black] (-0.1,0.5) node {$\bt$};
				\draw[fill=white] (0,0) circle (1pt);
				\draw[fill=white] (1,0) circle (1pt);
				\draw[fill=white] (0,1) circle (1pt);
				\draw[fill=white] (1,1) circle (1pt);
				\draw (0.5,0.5) node {$\mathcal F$};
				\end{tikzpicture}
			\end{center}
			\caption{Elementary quadrilateral, with vertices $x,u,v,y$ and face $\mathcal F$.}
			\label{fig_square}
		\end{figure}
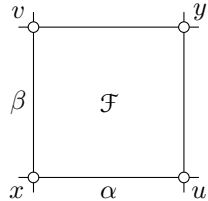
		Each ABS equation is written as the zero set of a polynomial of the variables at the vertices $x$, $u$, $v$, $y$:
		\begin{equation*}
			Q(x,u,v,y) = 0,
		\end{equation*}
		where $Q$ is affine linear in each of $x, u, v, y$. We can interpret this equation as a partial difference equation by identifying the vertices with a function $x_{l,m}$ in the following way:
		\begin{equation}\label{eq:PdiffE}
		(x,u, v,y)=(x_{l,m},x_{l+1,m},x_{l,m+1},x_{l+1,m+1}).
		\end{equation}
		We call such equations \textit{quad-equations}.
		
		The parameters $\al$ and $\bt$ (which may vary with the location of each vertex in the lattice) are associated with directions in the lattice, so that for example $\al$ is associated with the edge connecting $x,u$ as well as with the edge connecting $v,y$, and similarly $\bt$ is associated with both the edges $x,v$ and $u,y$.
		
		Such lattice equations are partial difference equations posed on quadrilaterals (see Figure \ref{fig_square}), where the solution on one vertex is given in terms of the values on the remaining three vertices. Three (or more) dimensional initial-value spaces therefore arise and we construct them as compactifications and regularisations of $\mathbb C^3$. To provide explicit constructions and results, we focus on specific examples of lattice equations from the ABS list. Our geometric construction leads us to new information about solutions, by providing new connections between distinct members of the ABS list via Miura transformation to a lattice equation first found in \cite{nijhoff2009soliton}. Our main results are summarized in Section \ref{sect_main}.
		\subsection{Main Results}\label{sect_main}
		Proposition \ref{prp:H3} provides a resolution of singularities for the ABS equation H3$_{\dl=0}$. Given $(l, m)$, The resolution is carried out on a quadrilateral, one vertex of which corresponds to this value on the lattice. In the proof of this result, we show that the map is completely regularised after blowing up four lines and four points.  Each exceptional variety that arises in a resolution is also considered under iteration of $l$ and $m$.
		
		Theorem \ref{thm:coalescence} uses Proposition \ref{prp:H3} to show that several ABS equations are mapped to one and the same equation \eqref{eq:mystery} under a previously unknown Miura transformation. This result provides a Miura transformation of H3, Q1, Q3, and A1 to a lattice equation first found in \cite{nijhoff2009soliton}, namely Equation \eqref{eq:mystery}.
		
		Proposition \ref{prp:reduction} focuses on Equation \eqref{eq:mystery} and provides a reduction to a QRT map.
		
		\subsection{Background}
		The study of initial-value spaces began with Okamoto's study of the classical Painlev\'e equations \cite{Oka1979}. Since the solutions become unbounded in domains containing poles, Okamoto extended the space of initial values to include such cases by compactifying it and regularising each resulting two-dimensional surface. Sakai extended this framework to discrete Painlev\'e equations. For discrete equations, it is convenient to take the compactification to be $\mathbb{P}^1\times\mathbb{P}^1$ (shorthand for $\mathbb C\mathbb{P}^1\times \mathbb C\mathbb{P}^1$).
		
		Discrete Painlev\'e equations are second-order equations and correspondingly, require a two dimensional space of initial values. It is known that it is necessary to blow-up $\mathbb{P}^1\times\mathbb{P}^1$ at 8 points to resolve the space of initial values, and the resulting surfaces were described by Sakai.
		
		Following this construction of two-dimensional initial value spaces of Painlev\'e equations, in this paper we construct an initial value space for several examples of higher dimensional discrete systems. In this higher dimensional setting, there exist algebraic sets of initial values where the equation becomes ill-defined. In the case of $d$-dimensional spaces of initial values where $d>2$, we refer to  the singular sets as {\em base varieties} in keeping with the terminology \cite{Dui}. Our focus lies more specifically on integrable partial difference equations which are self-consistent on $N$-dimensional cubic lattices \cite{nijhoff2001discrete}. Adler {\it et al} \cite{AdlBobSur2003,AdlBobSur2009a} showed how to classify such partial difference equations under certain conditions, and the list of equations they classified is known as the ABS list. We will be concentrating on concrete examples from the ABS list known as A1, H3, Q1, and Q3, which are given by
		\begin{subequations}
			\begin{align}
				\al\,(x+v)(u+y)-\bt\,(x+u)(v+y)-\dl^2\,\al\,\bt&=0\label{eq:A1},\\
				\al\,(x\,u+v\,y)-\bt(x\,v+u\,y)+\dl\,(\al^2-\bt^2)&=0,\\
				\al\,(x-v)(u-y)-\bt\,(x-u)(v-y)+\dl^2\,\al\,\bt(\al-\bt)&=0\label{eq:Q1},\\
				(\bt^2-\al^2)(x\,y+u\,v)+\bt\,(\al^2-1)(x\,u+v\,y)-\al\,(\bt^2-1)(x\,v+u\,y)&\,\nonumber\\
				-\dl^2(\al^2-\bt^2)(\al^2-1)(\bt^2-1)/(4\,\al\,\bt)&=0\label{eq:Q3},
			\end{align}
		\end{subequations}
		respectively.
		
		To solve a partial difference equation everywhere on a lattice we require an infinite number of initial values, and hence its initial-value space must be infinite-dimensional.  However, if we start in a generic sub-domain of the lattice containing a finite set of vertices connected by edges and consider resolution of singularities for this sub-domain (in particular, a single quadrilateral) we only require a finite number of initial values. This is explained in further detail in Section \ref{sect_inits}.
		
		The deep connection between integrable partial difference equations and discrete Painlev\'e equations was discovered by Nijhoff {\it et al} \cite{NP1991}. Many examples of reductions from ABS equations to discrete Painlev\'e equations have now been found \cite{GRSWC2005,OVQ2013}. A {\em geometric reduction} method was provided in \cite{JNS2014,JNS2015, JNS2016}, to find systematic reductions from ABS equations on $n$-dimensional cubes to discrete Painlev\'e equations with specific symmetry groups; these include $A_2^{(1)}+A_1^{(1)}$, $A_1^{(1)}+A_1^{(1)}$, and $A_4^{(1)}$. To illustrate the relationship between ABS equations and discrete Painlev\'e equations, we provide an example of a reduction here. This type of reduction is commonly called a \textit{staircase} or \textit{periodic} reduction, and is the type we will be using in this paper.
		
		The equation H3$_{\dl=0}$ (also known as the lattice modified Korteweg-de Vries equation) has a staircase reduction to the equation
		\begin{equation}\label{eq:qp2}
		\wb w\,w\,\ub w = \frac{z\,w-1}{w-z},
		\end{equation}
		where $w=w(z)$, $\wb w=w(q\,z)$ and $\ub w=w(z/q)$. A non-autonomous version of equation \eqref{eq:qp2} is referred to as the second $q$-discrete Painlev\'e equation, or $q\mathrm{P}_{\mathrm{II}}$ in the literature \cite{GRSWC2005}. 
		\begin{figure}[H]
			\begin{center}
				\begin{tikzpicture}[scale=2]
				\draw[-] (0,0.5) -- (0,2.5) node[right] {$ $};
				\draw[-] (0.5,0.5) -- (0.5,2.5) node[right] {$ $};
				\draw[-] (-0.5,2) -- (3.5,2) node[right] {$ $};
				\draw[-] (-0.5,1.5) -- (3.5,1.5) node[right] {$ $};
				\draw[-] (-0.5,1) -- (3.5,1) node[right] {$ $};
				\draw[-] (2,0.5) -- (2,2.5) node[right] {$ $};
				\draw[-] (2.5,0.5) -- (2.5,2.5) node[right] {$ $};
				\draw[-] (3,0.5) -- (3,2.5) node[right] {$ $};
				\draw[-] (1.5,0.5) -- (1.5,2.5) node[right] {$ $};
				\draw[-] (1,0.5) -- (1,2.5) node[right] {$ $};
				
				\draw[ultra thick] (0,2) -- (1,2) ;
				\draw[ultra thick] (1,1.5) -- (2,1.5) ;
				\draw[ultra thick] (2,1) -- (3,1) ;
				\draw[ultra thick] (1,2) -- (1,1.5) ;
				\draw[ultra thick] (2,1.5) -- (2,1) ;
				
				\draw[ultra thick] (0,2) -- (0,2.5) ;
				\draw[ultra thick] (3,0.5) -- (3,1) ;
				
				\draw[dashed] (0,2.5) -- (3.5,0.75) ;
				\draw[dashed] (-0.5,2.25) -- (3,0.5) ;
				
				\draw[->] (-0.5,0.2) -- (0,0.2) node[right] {$l$};
				\draw[->] (-0.7,0.5) -- (-0.7,1) node[above] {$m$};
				\end{tikzpicture}
			\end{center}
			\caption{Vertices of the staircase (those along each dotted line) are identified with each other.}
			\label{staircase0}
		\end{figure}
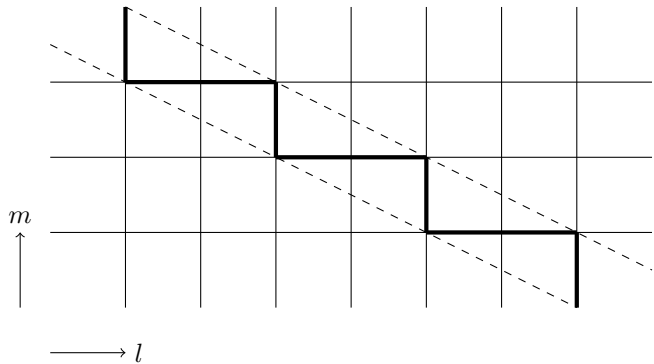
		The reduction is obtained by assuming that the solution $x_{l,m}$ of H3$_{\dl=0}$ stays constant along a line given by $l+2\,m=constant$ on the lattice, denoted by dotted lines in Figure \ref{staircase0}. The intersections of each dotted line with the lattice in Figure \ref{staircase0} coincide with points describing treads of a staircase, and so this type of reduction is commonly referred to as a \textit{staircase} reduction. Modifications of reductions of this type are also possible. For example, a periodic reduction given by $w_{l,m+1}=k/w_{l+1,m}$, for nonzero $k$ with $w_{l,m}$ is studied in Section \ref{sect:Q3}.
		\subsection{Plan of the Paper}\label{s:plan}
		The plan of the paper is as follows. In Section \ref{sect_inits}, we describe the initial conditions of lattice equations on quadrilaterals in a lattice and demonstrate the resolution of higher dimensional base varieties that occur in $\mathbb{P}^3$. We then show how to construct an appropriate regularised space of initial conditions analogous to the corresponding spaces known for the discrete Painlev\'e equations. In Section \ref{sect:Q3} we show how to use this resolution of singularities to find transformations of lattice equations and hence find new reductions of several examples of ABS equations.
		
		\section{Initial Values of Lattice Equations on Quad Graphs}\label{sect_inits}
		In this section, we start by giving notation for initial values of ABS equations on a lattice. We then explain how to carry out the resolution of singularities in spaces of dimension greater than two, with our main example being H3 from the ABS list. The resolution is carried out for initial values given on a generic quadrilateral of the lattice of independent variables described by a fixed vertex $(l, m)$.
		
		In general, lattice equations have an infinite-dimensional space of initial conditions. However, when we restrict our attention to a finite number of iterations on the lattice, it is sufficient to consider certain finite sets of initial data given on a bounded sub-domain of the lattice. In particular, when we iterate only once, initial conditions are only needed on a single quadrilateral: see Figure \ref{fig_square}. A quad-equation provides the value of the variable $y$ on a vertex of the quadrilateral when three initial conditions $x$, $u$, $v$ are given on the remaining 3 vertices. We denote the corresponding map on projective space as
		$$\phi_3:\mathbb{P}^3\longrightarrow\mathbb{P}^1, [x:u:v:1]\longmapsto [y:1].$$
		
		The lattice equation defining $y$ in terms of $x, u, v$ possesses indeterminancies where the image of $\phi_3$ becomes $[0:0]$, or equivalently, where $y$ becomes $\frac{0}{0}$. We refer to such subvarieties of $\mathbb{P}^3$ as \textit{base varieties} or \textit{base curves}. For the map $\phi_3$, we initially find 1-dimensional (codimension-2) base varieties in $\mathbb{P}^3$. In order to resolve these base varieties we perform a blow-up, replacing each base curve with a surface isomorphic to $\mathbb{P}^1\times\mathbb{P}^1$ -- a so-called \textit{exceptional plane}.
		
		For some base varieties, it may not be possible to view the entire curve in a single affine chart. For this reason we consider the map in homogeneous coordinates. Changing to homogeneous coordinates in $\mathbb{P}^3$ by taking
		$$[x:u:v:1]=[X/W:U/W:V/W:1]=[X:U:V:W],$$
		we take a base variety $B$ defined by the simultaneous vanishing of the functions $\lambda([X:U:V:W])$ and $\mu([X:U:V:W])$. The blow-up of $\mathbb{P}^3$ centred at $B$ is then the surface defined as
		\[\tilde{X}=
		\left\{([X:U:V:W],[\xi:\eta])\ \vline\ \lambda\,\xi-\mu\,\eta=0\right\}\subset\mathbb{P}^3\times\mathbb{P}^1.
		\]
		
		As in the 2-dimensional case, we have a new projective coordinate given by \linebreak$[\xi:\eta]=[\mu:\lm]$. We then cover the exceptional plane in affine charts, one using the coordinate $\xi/\eta$, and another $\eta/\xi$. In general, in order to resolve any codimension-$k>1$ subvariety which is the locus of equations $x_1=\dots=x_k=0$ in a space $X$, take $y_1,\dots,y_k$ as homogeneous coordinates of $\mathbb{P}^{k-1}$. Then the blow-up $\tilde{X}$ is the locus of the equations $x_i\,y_j=x_j\,y_i\ \forall\ i,j$, in the space $X\times\mathbb{P}^{k-1}$ \cite{Harr}.
		
		After blowing up these base varieties to exceptional planes (each isomorphic to $\mathbb{P}^1\times\mathbb{P}^1$), there may be singular varieties remaining on each plane. Upon resolving all singular varieties which may appear after all necessary blow-ups, we say the system is resolved.
		\subsection{Resolution of H3}
		We now focus on the explicit resolution of base varieties for the ABS equation H3$_{\dl=0}$. Setting $\dl=0$ and solving for a vertex $y$ on a quadrilateral as in Figure \ref{fig_square}, we can express H3 as
		\begin{equation}\label{eq_H3y}
		y=x\frac{\al\,u-\bt\,v}{\bt\,u-\al\,v}.
		\end{equation}
		To consider this in homogeneous coordinates we define $$[x:u:v:1]=:[X/W:U/W:V/W:1]=[X:U:V:W].$$
		In these coordinates, $y$ is given by
		$$y=\frac{X(\al\,U-\bt\,V)}{W(\bt\,U-\al\,V)}.$$
		Recall that $[0:0:0:0]$ is not defined in $\mathbb{P}^3$. We find that the map 
		\begin{equation}\label{labelphi3}
		\phi_3:\mathbb{P}^3\longrightarrow\mathbb{P}^1, [X:U:V:W]\longmapsto [y:1],
		\end{equation}
		has 4 base lines, given by
		\begin{align}
			U=V=& 0,\tag{B$_1$}\\
			X=\al\,V-\bt\,U=& 0, \tag{B$_2$}\\
			W=\bt\,V-\al\,U=& 0, \tag{B$_3$}\\
			W=X=& 0. \tag{B$_4$}
		\end{align}
		The relative position and intersections of these lines in $\mathbb{P}^3$ is shown in Figure \ref{img:H3P3}. The corners of the tetrahedron are labelled with the their corresponding coordinates in $\mathbb{P}^3$.
		
		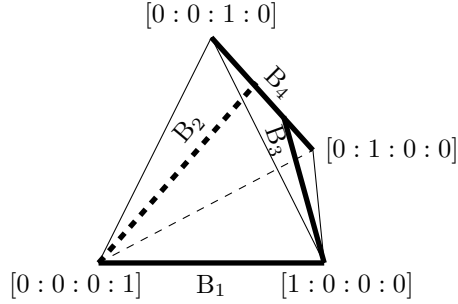
\begin{figure}[H]
			\begin{center}
				\begin{tikzpicture}[scale=3]
				\draw[-] (-0.5,0) -- (0,1);
				\draw[-] (0.5,0) -- (0,1);
				\draw[-] (0,1) -- (0.45,0.5);
				\draw[-] (0.5,0) -- (0.45,0.5);
				\draw[dashed] (-0.5,0) -- (0.45,0.5);
				\draw[line width=2pt] (-0.5,0) -- (0.5,0);
				\draw[line width=2pt] (0,1) -- (0.45,0.5);
				\draw[line width=2pt] (0.5,0) -- (0.32,0.65);
				\draw[line width=2pt,dashed] (-0.5,0) -- (0.2,0.8);
				\node[] at (-0.6,-0.1) {$[0:0:0:1]$};
				\node[] at (0.6,-0.1) {$[1:0:0:0]$};
				\node[] at (0,1.1) {$[0:0:1:0]$};
				\node[] at (0.8,0.5) {$[0:1:0:0]$};
				\node[rotate=0] at (0,-0.1) {B$_1$};
				\node[rotate=50] at (-0.1,0.6) {B$_2$};
				\node[rotate=-70] at (0.29,0.54) {B$_3$};
				\node[rotate=-45] at (0.3,0.8) {B$_4$};
				\end{tikzpicture}
				\caption{Base varieties of H3$_{\dl=0}$ in $\mathbb{P}^3$, base varieties drawn in bold, lines on rear faces are represented with dots.}
			\end{center}
			\label{img:H3P3}
		\end{figure}
		\begin{prp}\label{prp:H3}
			The map $\phi_3$ \eqref{labelphi3} is resolved after blowing up along 4 lines and 4 points.
		\end{prp}
		We start the resolution by resolving each of the base lines B$_1,\dots,$B$_4$, then considering their intersections. First, we resolve B$_1$ with the following change of variables.
		\begin{equation}
		\begin{cases}
		X_{11}&= X,\\
		U_{11}&=U,\\
		V_{11}&=\dfrac{V}{U},\\
		W_{11}&=W,
		\end{cases}
		\quad\mathrm{or}\quad
		\begin{cases}
		X&= X_{11},\\
		U&=U_{11},\\
		V&=U_{11}\,V_{11},\\
		W&=W_{11},
		\end{cases}
		\end{equation}
		and, upon setting $U_{11}=0$ to restrict the equation to the exceptional plane E$_1$, we find
		\begin{equation}
		y=\frac{\bt\,V_{11}\,X_{11}-\al\, X_{11}}{W_{11}(\al\,V_{11}-\bt)}.
		\end{equation}
		To cover all of E$_1$, we must also look in the other blow-up chart by considering the following change of variables.
		\begin{equation}
		\begin{cases}
		X_{12}&= X,\\
		U_{12}&=\dfrac{U}{V},\\
		V_{12}&=V,\\
		W_{12}&=W,
		\end{cases}
		\quad\mathrm{or}\quad
		\begin{cases}
		X&= X_{12},\\
		U&=U_{12}\,V_{12},\\
		V&=V_{12},\\
		W&=W_{12}.
		\end{cases}
		\end{equation}
		Upon setting $V_{12}=0$ to restrict the equation to the exceptional plane E$_1$, we find
		\begin{equation}
		y=\frac{\bt\,X_{12}-\al\, X_{12}\,U_{12}}{W_{12}(\al-\bt\,U_{12})}.
		\end{equation}
		In both charts we find $\phi_3$ is undefined where $X_{11}=X_{12}=X=0$ and $W_{11}=W_{12}=W=0$. However, in the original coordinates this corresponds to $[0:0:0:0]$, and is therefore not in the domain.
		
		Since $X$ and $W$ cannot vanish simultaneously on E$_1$, by taking the pair $$([u_1:1],[v_1:1])\in\mathbb{P}^1\times\mathbb{P}^1,$$ where both
		\begin{align*}
			[u_1:1]&=[X_{11}:W_{11}]=[X:W],\\
			[v_1:1]&=[V_{11}:1]=[1:U_{12}]=[v:u],
		\end{align*}
		hold, then $$([u_1:1],[v_1:1])\in\mathbb{P}^1\times\mathbb{P}^1$$ act as projective coordinates on E$_1\cong\mathbb{P}^1\times\mathbb{P}^1$.
		
		The map $\phi_3$ restricted to E$_1$ induces the map $\phi_{31}:\mathbb{P}^1\times\mathbb{P}^1\to\mathbb{P}^1$. We cover E$_1$ in affine coordinate charts. For example, in the chart $(u_1,v_1)=(X_{11}/W_{11},V_{11})$ we have
		\begin{equation}
		y=\frac{\al\,u_1-\bt\,u_1\,v_1}{\bt-\al\,v_1}.
		\end{equation}
		In this chart we see the base point $(u_1,v_1)=(0,\frac{\bt}{\al})$. Note that this point on E$_1$ actually corresponds to the intersection with the base line B$_2$. In the chart $(U_1,V_1):=(\frac{1}{u_1},\frac{1}{v_1})=(W_{12}/X_{12},U_{12})$, we have
		$$y=\frac{\al\,V_1-\bt}{\bt\,U_1\,V_1-\al\,U_1}.$$
		In this chart we see the base point $(U_1,V_1)=(0,\frac{\bt}{\al})$. This base point corresponds to the intersection of E$_1$ with B$_3$. These intersections are all the base varieties we find on E$_1$.
		
		Next we blow-up along B$_2$. Applying the change of variables
		\begin{equation}
		\begin{cases}
		X_{21}&= X,\\
		U_{21}&=U,\\
		V_{21}&=\dfrac{\al\,V-\bt\,U}{X},\\
		W_{21}&=W,
		\end{cases}
		\quad\mathrm{or}\quad
		\begin{cases}
		X&=X_{21},\\
		U&=U_{21},\\
		V&=\dfrac{X_{21}\,V_{21}+\bt\,U_{21}}{\al},\\
		W&=W_{21},
		\end{cases}
		\end{equation}
		and setting $X_{21}=0$, on the exceptional plane E$_2$ we find
		$$y=\frac{\bt^2-\al^2}{\al}\frac{U_{21}}{W_{21}}\frac{1}{V_{21}}.$$
		In the other blow-up chart, applying the change of variables
		\begin{equation}\label{refnextsection}
		\begin{cases}
		X_{22}&= \dfrac{X}{\al\,V-\bt\,U},\\
		U_{22}&=U,\\
		V_{22}&=\al\,V-\bt\,U,\\
		W_{22}&=W,
		\end{cases}
		\quad\mathrm{or}\quad
		\begin{cases}
		X&=X_{22}\,V_{22},\\
		U&=U_{22},\\
		V&=\dfrac{V_{22}+\bt\,U_{22}}{\al},\\
		W&=W_{22},
		\end{cases}
		\end{equation}
		and setting $V_{22}=0$, we find
		$$y=\frac{\bt^2-\al^2}{\al}\frac{U_{22}}{W_{22}}X_{22}.$$
		Once again, note that $U_{12}=U_{21}$ and $W_{12}=W_{21}$ cannot vanish simultaneously on E$_2$. Take the pair $([u_2:1],[v_2:1])\in\mathbb{P}^1\times\mathbb{P}^1$, such that $[u_2:1]=[U_{21}:W_{21}]=[U:W]$ and $[v_2,1]=[V_{21}:1]=[1:X_{22}]$, and use $([u_2:1],[v_2:1])$ as coordinates on E$_2\cong\mathbb{P}^1\times\mathbb{P}^1$. The map $\phi_3$ restricted to E$_2$ induces a map $\phi_{32}:\mathbb{P}^1\times\mathbb{P}^1\to\mathbb{P}^1$. 
		
		Similarly for B$_3$ and B$_4$, after a single blow-up we find on the exceptional planes E$_3$ and E$_4$ the maps,
		\begin{align*}
			&\phi_{33}:\mathbb{P}^1\times\mathbb{P}^1\longrightarrow\mathbb{P}^1,\\
			&([u_3:1],[v_3:1])=([X:V],[W:\bt\,V-\al\,U])\longmapsto[y:1]=[\al\,u_3:v_3(\al^2-\bt^2)],\\
			&\phi_{34}:\mathbb{P}^1\times\mathbb{P}^1\longrightarrow\mathbb{P}^1,\\
			&([u_4:1],[v_4:1])=([V:U],[W:X])\longmapsto[y:1]=[\bt\,u_4-\al:v_4(\al\,u_4-\bt)],
		\end{align*}
		respectively. Note that all base points on every exceptional plane correspond to the intersection of the exceptional plane with another base variety. These base points on the exceptional planes corresponding to intersections with other base varieties are each resolved after one blow-up. This proves the proposition.
		
		\section{Transformations of lattice equations from geometry}\label{sect:Q3}
		In this section, we consider iterations of exceptional varieties found in Section 2 on the lattice. This approach turns out to yield new transformations and reductions of ABS equations. 
		\begin{figure}[H]
			\begin{center}
				\begin{tikzpicture}
				\foreach \x in {-1,0,+1}	{
					\draw (-4,-3*\x) -- (4,-3*\x);
					\draw (-3*\x,-4) -- (-3*\x,4);
				}
				\foreach \x in {-1,0,+1}	{
					\foreach \y in {-1,0,+1}	{
						\draw[black,fill=white] (3*\x,3*\y) circle (2.5pt);
					}
				}
				\draw[black,fill=white] (-3,-3) node[above left] {$x_{l-1,m-1}$};
				\draw[black,fill=white] (0,-3) node[above left] {$x_{l,m-1}$};
				\draw[black,fill=white] (3,-3) node[above left] {$x_{l+1,m-1}$};
				
				\draw[black,fill=white] (-3,0) node[above left] {$x_{l-1,m}$};
				\draw[black,fill=white] (0,0) node[above left] {$x_{l,m}$};
				\draw[black,fill=white] (3,0) node[above left] {$x_{l+1,m}$};
				
				\draw[black,fill=white] (-3,3) node[above left] {$x_{l-1,m+1}$};
				\draw[black,fill=white] (0,3) node[above left] {$x_{l,m+1}$};
				\draw[black,fill=white] (3,3) node[above left] {$x_{l+1,m+1}$};
				
				\draw[black,fill=white] (-1.5,-1.5) node[] {$\mathcal{F}_{l,m}$};
				\draw[black,fill=white] (1.5,-1.5) node[] {$\mathcal{F}_{l+1,m}$};
				
				\draw[black,fill=white] (-1.5,1.5) node[] {$\mathcal{F}_{l,m+1}$};
				\draw[black,fill=white] (1.5,1.5) node[] {$\mathcal{F}_{l+1,m+1}$};
				\end{tikzpicture}
			\end{center}
			\caption{For each quadrilateral with vertices $(l, m)$, $(l+1, m)$, $(l, m+1)$, $(l+1, m+1)$, and corresponding initial values $x_{l,m}$, $x_{l+1, m}$, $x_{l,m+1}$, we denote the resolved initial value space corresponding to that quadrilateral by $\mathcal F_{l,m}$.} 
			\label{fig5:asf2d}
		\end{figure}
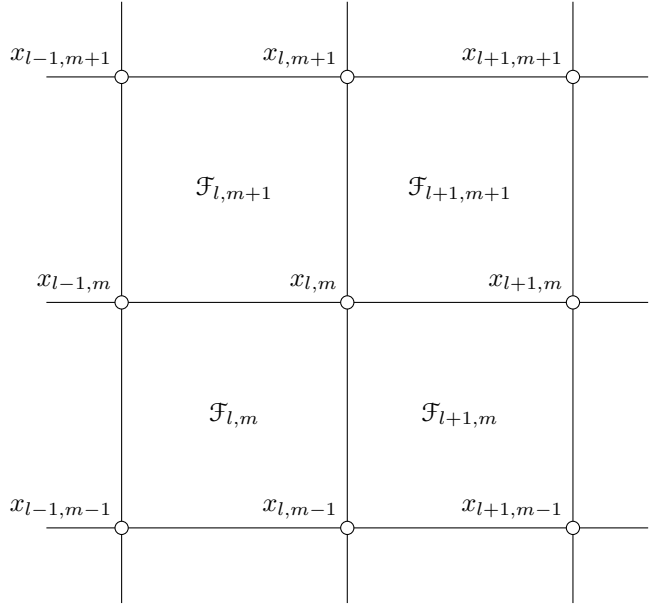
		In what follows, we fix $(l, m)$, and hence a quadrilateral, and carry out resolutions of an ABS equation in the space of initial values corresponding to this quadrilateral. As shown in Figure \ref{fig5:asf2d}, we denote the resolved space for the quadrilateral by $\mathcal F_{l,m}$.
		
		\subsection{Transformations arising from resolution of singularities}
		We prove the following theorem.
		\begin{thm}\label{thm:coalescence}
			The ABS equations H3$_{\dl=0}$, Q1$_{\dl}$, Q3$_{\dl=0}$, and A1$_{\dl=0}$ all have transformations to the equation
			\begin{equation}\label{eq:mystery}
			w_{l+1,m+1}=w_{l,m}\frac{(w_{l+1,m}+a)(w_{l,m+1}+b)}{(w_{l+1,m}+b)(w_{l,m+1}+a)},
			\end{equation}.
		\end{thm}
		The proof is provided case by case in the following subsections.
		\subsubsection{Transformations of H3$_{\dl=0}$}\label{proofH3}
		Consider H3$_{\dl=0}$ \eqref{eq_H3y}, which we have already resolved in Section 2. We focus on the exceptional plane E$_2$, which is the result of the blow-up given by \eqref{refnextsection}.
		
		The resolved space $\mathcal{F}_{l,m}$ of each quadrilateral contains a copy of this exceptional plane E$_2$. It is parameterised by $x_{l+1,m}$ and
		\begin{align}\label{defwh3}
			w_{l+1,m+1}:=\frac{x_{l,m}}{\al\,x_{l,m+1}-\bt\,x_{l+1,m}},
		\end{align}
		corresponding to the variable $X_{22}$ (see \eqref{refnextsection}).
		
		Using \eqref{eq_H3y}, we re-express $w_{l+1,m+1}$ in terms of initial values corresponding to neighbouring quadrilaterals. In particular, by iterating back to an earlier set of initial conditions. We find 
		\begin{align}\label{eq5:u11}
			w_{l+1,m+1}&=x_{l,m}\,(\bt\,x_{l,m}-\al\,x_{l-1,m+1})\,(\al\,x_{l,m}-\bt\,x_{l+1,m-1})\,/\notag\\
			&\qquad\left(\al\,x_{l-1,m}\,(\al\,x_{l,m}-\bt\,x_{l-1,m+1})\,(\al\,x_{l,m}-\bt\,x_{l+1,m-1})\right.\notag\\
			&\qquad\qquad\left.-\bt\,x_{l,m-1}\,(\bt\,x_{l,m}-\al\,x_{l-1,m+1})\,(\bt\,x_{l,m}-\al\,x_{l+1,m-1})\right).
		\end{align}
		Using \eqref{defwh3}, we make the following substitutions
		\begin{subequations}
			\begin{align}
				x_{l+1,m-1}=\frac{\al\,x_{l,m}\,w_{l+1,m}-x_{l,m-1}}{\bt\,w_{l+1,m}},\\
				x_{l-1,m+1}=\frac{x_{l-1,m}+\bt\,x_{l,m}\,w_{l,m+1}}{\al\,w_{l,m+1}}.
			\end{align}
		\end{subequations}
		Equation \eqref{eq5:u11} becomes
		\begin{equation}
		w_{l+1,m+1}=\frac{x_{l,m}}{\bt\,x_{l-1,m}-\al\,x_{l,m-1}+(\bt^2-\al^2)(w_{l,m+1}-w_{l+1,m})\,x_{l,m}}.
		\end{equation}
		Finally, using \eqref{eq_H3y} to rewrite $x_{l,m}$ in terms of $x_{l-1,m-1},x_{l-1,m},x_{l,m-1}$ and making the substitution for $w_{l,m}$ using \eqref{eq_H3y}, we find
		\begin{equation}\label{firstn}
		w_{l+1,m+1}=\frac{w_{l,m}}{(\bt^2-\al^2)(w_{l,m+1}-w_{l+1,m})\,w_{l,m}+1}.
		\end{equation}
		To the best of our knowledge, this transformation of H3$_{\dl=0}$ using \eqref{defwh3} does not appear in the literature.
		
		Exactly the same procedure applied to the exceptional plane resulting from the blow-up of B$_1$ would have led to the following transformation
		\begin{equation}
		u_{l+1,m+1}=\frac{x_{l,m+1}}{x_{l+1,m}}.
		\end{equation}
		This gives the equation
		\begin{equation}\label{eq:buh}
		u_{l+1,m+1}=u_{l,m}\frac{(u_{l+1,m}-\frac{\bt}{\al})(u_{l,m+1}-\frac{\al}{\bt})}{(u_{l+1,m}-\frac{\al}{\bt})(u_{l,m+1}-\frac{\bt}{\al})}.
		\end{equation}
		This particular transformation is known and was found in \cite{nijhoff2009soliton}. This proves Theorem \ref{thm:coalescence} for the case H3$_{\dl=0}$.
		
		\subsubsection{Transformations of Q3$_{\dl=0}$}\label{proofQ3}
		Consider Q3$_{\dl=0}$, given by \eqref{eq:Q3}. This equation has 3 base varieties B$_1$, B$_2$, and B$_3$, given by
		\begin{align}
			\al\,X-U=\bt\,X-V=\,0,\tag{B$_1$}\\
			X-\al\,U=X-\bt\,V=\,0, \tag{B$_2$}\\
			W=\bt\,(\al^2-1)\,X\,U-\al\,(\bt^2-1)\,X\,V+(\bt^2-\al^2)\,U\,V=\,0, \tag{B$_3$}
		\end{align}
		where $[x:u:v:1]=:[X:U:V:W].$
		
		Consider B$_2$. To resolve B$_2$ we define
		\begin{align*}
			\begin{cases}
				X_{11}&=X,\\
				U_{11}&=\dfrac{X-\al\,U}{X-\bt\,V},\\
				V_{11}&=X-\bt\,V,\\
				W_{11}&=W,
			\end{cases}\nonumber
			\quad \mathrm{or}\quad
			\begin{cases}
				X&=X_{11},\\
				U&=\dfrac{X_{11}-U_{11}\,V_{11}}{\al},\\
				V&=\dfrac{X_{11}-V_{11}}{\bt}.
			\end{cases}
		\end{align*}
		B$_2$ is replaced by an exceptional plane E$_2$ which is parameterised by $X_{11}$ and $U_{11}$. Define 
		\begin{equation}\label{eq:blowUpQ3}
		w_{l+1,m+1}:=\frac{x_{l,m}-\al\,x_{l+1,m}}{x_{l,m}-\bt\,x_{l,m+1}}.
		\end{equation}
		Iterating into neighbouring quadrilaterals, we find
		\begin{align}\label{flightscannerco2}
			\nonumber w_{l+1,m+1}=&((\al ^2-1)\,(x_{1,m}-\bt\,x_{l,m-1})\, (\bt\,x_{l,m}-\al\,x_{l+1,m-1}) \\
			\nonumber &\qquad ((\al^2-\bt ^2)\,x_{l-1,m}-\bt\,(\al ^2-1)\, x_{l-1,m+1}+\al\,(\bt ^2-1)\,x_{l,m}))\\
			\nonumber &\qquad/((\bt^2-1)\,(x_{l,m}-\al\,x_{l-1,m})\,(\al\,x_{l,m}-\bt\,x_{l-1,m+1})\\
			&\qquad(\bt\,(\al^2-1)\,x_{l,m}+(\bt^2-\al ^2)\,x_{l,m-1}-\al\,(\bt ^2-1)\,x_{l+1,m-1})).
		\end{align}
		We eliminate $x_{l+1,m-1}$, $x_{l-1,m+1}$ by using \eqref{eq:blowUpQ3}, rewritten as
		\begin{subequations}
			\begin{align}
				x_{l+1,m-1}&=\frac{x_{l,m-1}-w_{l+1,m}\,x_{l,m-1}+\bt\,w_{l+1,m}\, x_{l,m}}{\al},\\
				x_{l-1,m+1}&=\frac{w_{l,m+1}\,x_{l-1,m}-x_{l-1,m}+\al\,x_{l,m}}{\bt\,w_{l,m+1}},
			\end{align}
		\end{subequations}
		and eliminate $x_{l,m}$, $x_{l-1,m}$, and $x_{l,m-1}$ by using \eqref{eq:Q3} and \eqref{eq:blowUpQ3} in sequence to arrive at
		$$w_{l+1,m+1}=w_{l,m}\frac{(w_{l+1,m}-1)\,(\al^2-1-(\bt^2-1)\,w_{l,m+1})}{(w_{l,m+1}-1)\,(\al^2-1-(\bt^2-1)\,w_{l+1,m})}.$$
		
		Defining the parameter
		$$r:=\frac{\al^2-1}{\bt^2-1},$$
		we find the following lattice equation for $w_{l,m}$,
		\begin{equation}\label{eq:lowQ3}
		w_{l+1,m+1}=w_{l,m}\frac{(w_{l+1,m}-1)\,(w_{l,m+1}-r)}{(w_{l,m+1}-1)\,(w_{l+1,m}-r)}.
		\end{equation}
		Hence we have proven Theorem \ref{thm:coalescence} for the case Q3$_{\dl=0}$.
		
		Other exceptional planes lead to additional transformations. Following this same procedure using E$_1$ leads us to introduce the variable $v_{l,m}$, where
		$$v_{l+1,m+1}=\frac{\al\,x_{l,m}-x_{l+1,m}}{\bt\,x_{l,m}-x_{l,m+1}}.$$
		This a transformation of Q3$_{\dl=0}$ to
		\begin{equation}\label{eq:otherLatticeEqn}
		v_{l+1,m+1}=v_{l,m}\frac{(v_{l+1,m}-\frac{1}{r_1})(v_{l,m+1}-r_1\,r_2)}{(v_{l,m+1}-\frac{1}{r_1})(v_{l+1,m}-r_1\,r_2)},
		\end{equation}
		where the values of $r_1$ and $r_2$ are given by
		$$r_1=\frac{\al}{\bt},\ r_2=\frac{\al^2-1}{\bt^2-1}.$$ This provides a second case of Theorem \ref{thm:coalescence} for Q3$_{\dl=0}$.
		
		\subsubsection{Transformations of Q1$_\dl$}\label{proofQ1}
		Consider the ABS equation Q1$_\dl$ given by \eqref{eq:Q1}.
		Resolution of base varieties leads to the transformations
		\begin{subequations}
			\begin{align}
				u_{l,m}=-\frac{x_{l,m}-x_{l+1,m}+\al\,\dl}{x_{l,m}-x_{l,m+1}+\bt\,\dl},\\ v_{l,m}=-\frac{x_{l,m}-x_{l+1,m}-\al\,\dl}{x_{l,m}-x_{l,m+1}-\bt\,\dl}.
			\end{align}
		\end{subequations}
		Following the same procedures as in earlier sections, we are led to the equations
		\begin{subequations}\label{A1variables}
			\begin{align}
				u_{l+1,m+1}&=u_{l,m}\frac{(u_{l+1,m}+1)(u_{l,m+1}+\frac{\bt}{\al})}{(u_{l+1,m}+\frac{\bt}{\al})(u_{l,m+1}+1)},\\
				v_{l+1,m+1}&=v_{l,m}\frac{(v_{l+1,m}+1)(v_{l,m+1}+\frac{\al}{\bt})}{(v_{l+1,m}+\frac{\al}{\bt})(v_{l,m+1}+1)}.
			\end{align}
		\end{subequations}
		Each these equations is of the form \eqref{eq:mystery}, proving Theorem \ref{thm:coalescence} for the case Q1$_{\dl}$.
		\subsubsection{Transformations of A1$_\dl$}\label{proofA1}
		Consider the ABS equation A1$_{\dl}$, given by \eqref{eq:A1}. Solving for the vertex $y$ and considering the equation as a map $\mathbb{P}^3\to\mathbb{P}^1$ we find three base curves. We will focus here on those which are visible in all affine charts. These are the base lines B$_1$ and B$_2$, given by
		\begin{align}
			x+u+\al\,\dl=x+v+\bt\,\dl&=0,\tag{B$_1$}\\
			x+u-\al\,\dl=x+v-\bt\,\dl&=0.\tag{B$_2$}
		\end{align}
		Following the procedure outlined above, define the variables $u_{l,m}$, $v_{l,m}$ from the exceptional planes E$_1$ and E$_2$ respectively, so that
		\begin{subequations}
			\begin{align}
				u_{l,m}&=-\frac{x_{l,m}+x_{l+1,m}+\al\,\dl}{x_{l,m}+x_{l,m+1}+\bt\,\dl},\\
				v_{l,m}&=-\frac{x_{l,m}+x_{l+1,m}-\al\,\dl}{x_{l,m}+x_{l,m+1}-\bt\,\dl}.
			\end{align}
		\end{subequations}
		Unlike  the earlier cases of Q1$_\dl$, Q3$_{\dl=0}$ and H3$_{\dl=0}$, we find a system of equations for $u_{l,m},v_{l,m}$:
		\begin{subequations}\label{eq:sys1}
			\begin{align}
				u_{l+1,m+1}=u_{l,m}\frac{(v_{l+1,m}+1)(\bt\,v_{l,m+1}+\al)}{(v_{l,m+1}+1)(\bt\,v_{l+1,m}+\al)},\\
				v_{l+1,m+1}=v_{l,m}\frac{(u_{l+1,m}+1)(\bt\,u_{l,m+1}+\al)}{(u_{l,m+1}+1)(\bt\,u_{l+1,m}+\al)}.
			\end{align}
		\end{subequations}
		Alternatively, using \eqref{eq:A1} to show that 
		\begin{equation}
		\frac{\bt\,v_{l,m+1}+\al}{\bt\,v_{l+1,m}+\al}=\frac{v_{l,m}\,v_{l,m+1}}{u_{l,m}\,u_{l,m+1}}\frac{\bt\,u_{l,m+1}+\al}{\bt\,u_{l+1,m}+\al},
		\end{equation}
		we find
		\begin{subequations}\label{eq:sys2}
			\begin{align}
				u_{l+1,m+1}=v_{l,m}\frac{(v_{l+1,m}+1)(\al\,u_{l,m+1}^{-1}+\bt)}{(v_{l,m+1}^{-1}+1)(\bt\,u_{l+1,m}+\al)},\\
				v_{l+1,m+1}=u_{l,m}\frac{(u_{l+1,m}+1)(\al\,v_{l,m+1}^{-1}+\bt)}{(u_{l,m+1}^{-1}+1)(\bt\,v_{l+1,m}+\al)}.
			\end{align}
		\end{subequations}
		When $\dl=0$, $u_{l,m}$ and $v_{l,m}$ coincide as shown by \eqref{A1variables}. In this case, \eqref{eq:sys1} and \eqref{eq:sys2} both become
		\begin{equation}\label{eq:sys3}
		u_{l+1,m+1}=u_{l,m}\frac{(u_{l+1,m}+1)(u_{l,m+1}+\frac{\al}{\bt})}{(u_{l,m+1}+1)(u_{l+1,m}+\frac{\al}{\bt})},
		\end{equation}
		which is again a form of \eqref{eq:mystery}. This proves Theorem \ref{thm:coalescence} for the case A1$_{\dl=0}$.
		
		The results of Sections \ref{proofH3}, \ref{proofQ3}, \ref{proofQ1}, \ref{proofA1} collectively prove Theorem \ref{thm:coalescence}.
		\subsection{Reduction}
		Scaling symmetries show that Equation \eqref{eq:mystery} only depends on one parameter. To see this, consider the trivial scaling $w_{l,m}\mapsto a\,w_{l,m}$, $a\ne 0$, for all $l,m$. If we define $\gm=b/a$ then we have the single parameter equation
		\begin{equation}\label{gmw}
		w_{l+1,m+1}=w_{l,m}\frac{(w_{l+1,m}+1)\,(w_{l,m+1}+\gm)}{(w_{l+1,m}+\gm)\,(w_{l,m+1}+1)}.
		\end{equation}
		We now consider reductions of this equation.
		\begin{prp}\label{prp:reduction}
			The equation \eqref{gmw} has a periodic reduction to a QRT map given by
			\begin{align}\label{eq:refatend}
				y_{n+1}\,y_{n-1}&=\frac{k}{\gm}\frac{(y_n+\gm)\,(y_n+k)}{(y_n+1)\,(y_n+k/\gm)}.
			\end{align}
		\end{prp}
		\begin{proof}
			Assuming that for some $k$ the function $w_{l,m}$ satisfies the periodic condition $w_{l,m+1}=k/w_{l+1,m}$, we find from \eqref{gmw} 
			\begin{align*}
				\frac{k}{w_{l+2,m}}&=w_{l,m}\frac{(w_{l+1,m}+1)\,(\frac{k}{w_{l+1,m}}+\gm)}{(w_{l+1,m}+\gm)\,(\frac{k}{w_{l+1,m}}+1)},\\
			\end{align*}
			and hence,
			\begin{align*}
				w_{l+2,m}&=\frac{k}{\gm\,w_{l,m}}\frac{(w_{l+1,m}+\gm)\,(w_{l+1,m}+k)}{(w_{l+1,m}+1)\,(w_{l+1,m}+k/\gm)}.
			\end{align*}
			Finally, letting $n:=l+1$, $m=0$, and defining $y_n:=w_{n,0}$ we are led to \eqref{eq:refatend}. This proves the Proposition.
		\end{proof}
		The Equation \eqref{eq:refatend} is a QRT map with surface type $A_3^{(1)}$. We note that this equation can be deautonomised to an equation often referred to $q\mathrm{P}_{\mathrm{VI}}$ in the literature \cite{Sak2001}.
		
		\section{Conclusion}\label{sect_conc}
		In this paper, we carried out a resolution of singularities of the initial value space of a lattice equation for the first time. This creates a framework that naturally leads to transformations of variables that are used for resolution. This framework turns out to be useful in other ways. We have applied the framework to find new transformations of ABS equations to a single lattice equation. Moreover, we were led to natural reductions to ordinary difference equations. 
		
		In this paper we considered $A1_{\dl=0}$, H3$_{\dl=0}$, $Q1_{\dl}$, and Q3$_{\dl=0}$. We found a Miura transformation in each case that transformed the respective equation to a single lattice equation, namely \eqref{eq:mystery}. We also found reductions to the discrete Painlev\'e equation of surface type $A_3^{(1)}$.  The composition of the Miura transformation with the reduction of \eqref{eq:mystery} implies new reductions of the lattice equations  A1$_{\dl=0}$, H3$_{\dl=0}$, Q1$_{\dl}$, Q3$_{\dl=0}$.  
		
		Open questions include the actions of such transformations and reductions on special solutions and how to use these new reductions to find new Lax pairs for Painlev\'e equations. There is a deep relationship between linearized action on initial-value space and algebraic entropy, which remains an open question for lattice equations.
		
		\section*{Acknowledgements} The authors would like to thank T. Tsuda, L. Paunescu and C. Viallet for very helpful discussions. 
		\newpage
		
		\bibliographystyle{plain}

\begin{thebibliography}{10}
			
			\bibitem{AdlBobSur2003}
			V.E. Adler, A.I. Bobenko, and Y.B. Suris.
			\newblock Classification of integrable equations on quad-graphs. the
			consistency approach.
			\newblock {\em Communications in Mathematical Physics}, 2003.
			
			\bibitem{AdlBobSur2009a}
			V.E. Adler, A.I. Bobenko, and Y.B. Suris.
			\newblock Discrete nonlinear hyperbolic equations. classification of integrable
			cases.
			\newblock {\em Functional Analysis and Its Applications}, 2009.
			
			\bibitem{Dui}
			J.J. Duistermaat.
			\newblock {\em Discrete Integrable Systems: QRT Maps and Elliptic Surfaces}.
			\newblock Springer Monographs in Mathematics, 2010.
			
			\bibitem{GRSWC2005}
			B.~Grammaticos, A.~Ramani, J.~Satsuma, R.~Willox, and A.S. Carstea.
			\newblock Reductions of integrable lattices.
			\newblock {\em Journal of Nonlinear Mathematical Physics}, 2005.
			
			\bibitem{Harr}
			J.~Harris.
			\newblock {\em Algebraic Geometry}.
			\newblock Springer Monographs in Mathematics, 1992.
			
			\bibitem{JNS2014}
			N.~Joshi, N.~Nakazono, and Y.~Shi.
			\newblock Geometric reductions of {ABS} equations on an n-cube to discrete
			{P}ainlev{\'e} systems.
			\newblock {\em Journal of Physics A}, 47, 2014.
			
			\bibitem{JNS2015}
			N.~Joshi, N.~Nakazono, and Y.~Shi.
			\newblock Lattice equations arising from discrete {P}ainlev{\'e} systems ({I}):
			{$(A_2+A_1)^{(1)}$} and {$(A_1+A'_1)^{(1)}$} cases.
			\newblock {\em Journal of Mathematical Physics}, 56, 2015.
			
			\bibitem{JNS2016}
			N.~Joshi, N.~Nakazono, and Y.~Shi.
			\newblock Lattice equations arising from discrete {P}ainlev{\'e} systems. {II}.
			{$A^{(1)}_4$} case.
			\newblock {\em Journal of Physics A. Mathematical and Theoretical.}, 2016.
			
			\bibitem{KNY2016}
			K.~Kajiwara, M.~Noumi, and Y.~Yamada.
			\newblock Geometric aspects of {P}ainlev{\'e} equations.
			\newblock {\em arXiv:1509.08186v7}, 2016.
			
			\bibitem{nijhoff2009soliton}
			F.W. Nijhoff, J.~Atkinson, and J.~Hietarinta.
			\newblock Soliton solutions for {ABS} lattice equations: {I}. {C}auchy matrix
			approach.
			\newblock {\em Journal of Physics A: Mathematical and Theoretical},
			42(40):404005, 2009.
			
			\bibitem{NP1991}
			F.W. Nijhoff and V.G. Papageorgiou.
			\newblock Similarity reductions of integrable lattices and discrete analogues
			of {P}ainlev\'e {II} equation.
			\newblock {\em Physics Letters A}, 153, 1991.
			
			\bibitem{nijhoff2001discrete}
			F.W. Nijhoff and A.J. Walker.
			\newblock The discrete and continuous painlev{\'e} vi hierarchy and the garnier
			systems.
			\newblock {\em Glasgow Mathematical Journal}, 43(A):109--123, 2001.
			
			\bibitem{Oka1979}
			K.~Okamoto.
			\newblock {S}ur les feuilletages associ\'es aux equation du second ordre \`a
			points critiques fixes de {P. P}ainlev\'e.
			\newblock {\em Japanese Journal of Mathematics}, 1979.
			
			\bibitem{OVQ2013}
			C.M. Ormerod, P.~Van der Kamp, and G.R.W. Quispel.
			\newblock Discrete {P}ainlev{\'e} equations and their lax pairs as reductions
			of integrable lattice equations.
			\newblock {\em Journal of Physics A}, 46, 2013.
			
			\bibitem{Sak2001}
			H.~Sakai.
			\newblock Rational surfaces associated with affine root systems and geometry of
			the {P}ainlev\'e equations.
			\newblock {\em Communications in Mathematical Physics}, 2001.
			
		\end{thebibliography}

	\end{document}